\documentclass[onecolumn,nofootinbib,superscriptaddress,notitlepage,longbibliography,thightenlines,10pt]{revtex4-1}    

\usepackage[a4paper,hmargin=1in,vmargin=1.16in]{geometry}

\usepackage[ colorlinks = true,
             linkcolor = blue,
             urlcolor  = blue,
             citecolor = red,
             anchorcolor = green,
]{hyperref}

\usepackage{amssymb,amsmath,amsthm}
\usepackage[T1]{fontenc}
\usepackage{framed}
\usepackage[mathscr]{euscript}
\usepackage{microtype}

\newcommand{\cL}{\mathcal{L}}

\newcommand{\cP}{\mathcal{P}}
\newcommand{\cQ}{\mathcal{Q}}

\newcommand{\cT}{\mathcal{T}}

\newcommand{\cX}{\mathcal{X}}

\newcommand{\sH}{\mathscr{H}}

\newcommand{\bN}{\mathbb{N}}

\newcommand{\bC}{\mathbb{C}}

\newcommand{\olin}[1]{\cL(#1)}

\newcommand{\opos}[1]{\cP(#1)}

\newcommand{\Hmin}{H_{\mathrm{min}}}

\newcommand{\nbox}[2]{\hspace{#2pt} \mbox{#1} \hspace{#2pt}}
\newcommand{\ket}[1]{| \hspace{1pt} #1 \rangle}
\newcommand{\ketbrad}[2]{| \hspace{1pt} #1 \rangle \langle #2 \hspace{1pt} |}
\newcommand{\ketbra}[1]{\ketbrad{#1}{#1}}
\newcommand{\braket}[2]{\langle #1 \hspace{1pt} | \hspace{1pt} #2 \rangle}
\newcommand{\bra}[1]{\langle #1 \hspace{1pt} |}
\newcommand{\abs}[1]{| #1 |}

\newcommand{\tr}[1]{\mathrm{tr}(#1)}
\newcommand{\Tr}{\mathrm{tr}}
\newcommand{\id}{1}

\newcommand{\norm}[2][]{#1\| #2 #1\|}
\newcommand{\proj}[1]{|#1 \rangle\!\langle #1 |}

\newcommand{\pos}{{\sl pos}}

\newcommand{\dens}{{\cal S}}
\newcommand{\eps}{\varepsilon}
\newcommand{\pwin}{p_{\textnormal{win}}}
\newcommand{\pg}{p_{\textnormal{guess}}}
\newcommand{\BB}{\text{\rm\tiny BB84}}

\def\01{\{0,1\}}

\newtheorem{lemma}{Lemma}
\newtheorem{definition}{Definition}
\newtheorem{proposition}[lemma]{Proposition}
\newtheorem{corollary}[lemma]{Corollary}
\newtheorem{theorem}[lemma]{Theorem}

\begin{document}

\title{{A Monogamy-of-Entanglement Game With Applications to Device-Independent Quantum Cryptography}}

 \author{Marco Tomamichel}                        
 \email{cqtmarco@nus.edu.sg}
 \affiliation{Centre for Quantum Technologies, National University of Singapore}
 \author{Serge Fehr}                   
 \email{serge.fehr@cwi.nl}
 \affiliation{Centrum Wiskunde \& Informatica (CWI), Amsterdam, The Netherlands}
 \author{J\k{e}drzej Kaniewski}
 \author{Stephanie Wehner}
 \affiliation{Centre for Quantum Technologies, National University of Singapore}


\begin{abstract}
  We consider a game in which two separate laboratories collaborate to prepare a quantum system and are then asked to guess the outcome of a
  measurement performed by a third party in a random basis on that system. 
Intuitively, by the uncertainty principle and the monogamy of entanglement, the probability that {\em both} players simultaneously succeed in guessing the outcome 
correctly is bounded. 
We are interested in the question of how
  the success probability scales when many such games are performed in parallel. 
 We show that any strategy that maximizes the probability to win every game individually is
  also optimal for the parallel repetition of the game.
  Our result implies that the optimal guessing probability can be achieved without the use of entanglement.
  
  We explore several applications of this result. First, we show that it implies
  security for standard BB84 quantum key distribution when the receiving party uses {\em fully untrusted measurement devices}, i.e.\ we show that BB84 is one-sided device independent.
  Second, we show how our result can be used to prove security
  of a one-round position-verification scheme. 
  Finally, we generalize a well-known uncertainty relation
  for the guessing probability to quantum side information.
\end{abstract}

\maketitle

\section{Introduction}

Apart from their obvious entertainment value, games among multiple (competing) players often provide an intuitive way to understand complex problems.
For example, we may understand Bell inequalities in physics~\cite{bell64}, or interactive proofs in computer science~\cite{benor:ip}, as a game played 
by a referee against multiple provers~\cite{feige:mip,cleve:nonlocal}. 
Here we investigate a simple quantum multiplayer game whose analysis enables us to tackle several open questions in quantum cryptography. 

\subsection{Monogamy Game}

We consider a game played among three parties: Alice, Bob and Charlie (these players should be seen as operating in three different laboratories). In this game,
Alice takes the role of a referee and is assumed to be honest whereas Bob and Charlie form a team
determined to beat Alice.
A monogamy-of-entanglement game $\mathsf{G}$ consists of a list of 
measurements, $\mathcal{M}^\theta = \{F_x^\theta\}_{x \in \cal X}$, indexed by $\theta \in \Theta$, on a $d$-dimensional 
quantum system.

\begin{description}
    \item[Preparation Phase] Bob and Charlie agree on a \emph{strategy} and prepare an arbitrary 
    quantum state $\rho_{ABC}$, where $\rho_A$ has dimension $d$. They pass $\rho_A$
    to Alice and hold on to $\rho_B$ and $\rho_C$, respectively. After this phase, Bob
    and Charlie are no longer allowed to communicate.
\vspace{-0.3cm}
    \item[Question Phase] Alice chooses $\theta \in \Theta$ uniformly at random and measures $\rho_A$ using $\mathcal{M}^\theta$ to obtain the \emph{measurement outcome}, $x \in \cal X$. She then announces $\theta$ to Bob and Charlie.
\vspace{-0.3cm}
    \item[Answer Phase] Bob and Charlie independently form a guess of $x$ by performing a measurement (which may depend on $\theta$) on their respective shares of the quantum state.
\vspace{-0.3cm}
    \item[Winning Condition] The game is won if \emph{both} Bob and Charlie guess $x$ correctly.
\end{description}

From the perspective of classical information processing, our game may appear somewhat trivial\,---\,after all, if Bob and Charlie were to provide some classical information $k$
to Alice who would merely apply a random function $f_\theta$, they could predict the value
of $x = f_{\theta}(k)$ perfectly from $k$ and $\theta$. In quantum mechanics, however, 
the well-known uncertainty principle~\cite{heisenberg27} places a limit on how well observers can predict the outcome $x$ of incompatible measurements.

To exemplify this, we will in the following focus on the game $\mathsf{G}_{\BB}$ in which Alice measures a qubit 
in one of the two BB84 bases~\cite{bb84} to obtain a bit $x \in \01$ and use $\pwin(\mathsf{G}_{\BB})$ to denote the probability that Bob and Charlie win, maximized over all strategies. (A strategy is comprised of a tripartite state $\rho_{ABC}$, and, for each $\theta \in \Theta$, a measurement
on $B$ and a measurement on $C$.)
Then, if Bob and Charlie are restricted to classical memory (i.e., they are not entangled with Alice),
it is easy to see that they win the game with an (average) probability of at most 
$1/2 + 1/(2\sqrt{2}) \leq \pwin(\mathsf{G}_{\BB})$.%
\footnote{For example, this follows from a proof of an entropic uncertainty relation by 
Deutsch~\cite{deutsch83}.}

In a fully quantum world, however, uncertainty is not quite the end of the story as indeed Bob and Charlie are allowed to have a \emph{quantum} memory. To illustrate
the power of such a memory, consider the same game played just between Alice and Bob. As Einstein, Podolsky and Rosen famously observed~\cite{epr35}: If $\rho_{AB}$ is a maximally entangled 
state, then once Bob learns Alice's choice of measurement $\theta$, he can perform an adequate measurement on his share of the state to obtain $x$ himself. That is, there exists a strategy
for Bob to guess $x$ perfectly. Does this change when we add the extra player, Charlie? 
We can certainly be hopeful as it turns out that quantum entanglement is ``monogamous''~\cite{terhal04} in the sense that the more entangled Bob is with Alice, the less entangled
Charlie can be. In the extreme case where $\rho_{AB}$ is maximally entangled, even if Bob can guess $x$ perfectly every time, Charlie has to resort to making an uninformed random guess. As both of them have to be
correct in order to win the game, this strategy turns out to be worse than optimal.

An analysis of this game thus requires a tightrope walk between uncertainty on the one hand, and the monogamy of entanglement on the other. 
The following result is a special case of our main result (which we explain further down); a slightly weaker bound had been derived in~\cite{chandran09}, and the exact value had first been proven by Christandl and Schuch~\cite{christandlpc10}.%
\footnote{However, neither the techniques from~\cite{chandran09} nor from~\cite{christandlpc10} work for parallel repetitions. }

\begin{itemize}
  \item \textbf{Result (informal):}
  We find $\pwin(\mathsf{G}_{\BB}) = 1/2 + 1/(2\sqrt{2}) \approx 0.85$. Moreover, this value can be achieved when 
  Bob and Charlie have a classical memory only.
\end{itemize}
Interestingly, we thus see that monogamy of entanglement wins out entirely, canceling the power of Bob and Charlie's quantum memory -- the optimal winning probability can be achieved without any entanglement at all. In fact, this strategy results in a higher success probability than the one in which Bob is maximally entangled with Alice and Charlie is classical. In such a case the winning probability can be shown to be at 
most $1/2$. 
In spirit, this result is similar to (but not implied by) recent results obtained in the study of non-local games where the addition of one or more extra parties cancels the advantage coming from the use of entanglement~\cite{ito12}.

To employ the monogamy game for quantum cryptographic purposes, we need to understand what happens if we play the same game $\mathsf{G}$ $n$ times in parallel. 
The resulting game, $\mathsf{G}^{\times n}$, requires both Bob and Charlie to
guess the entire \emph{string} $x = x_1 \ldots x_n$ of measurement outcomes, where
$x_j$, $j \in [n]$, is generated by measuring $\rho_{A_j}$ ($\rho_{A_j}$ is the quantum state provided by Bob and Charlie in the $j$-th round of the game) in the basis $\mathcal{M}^{\theta_j}$, and
$\theta_j \in \Theta$ is chosen uniformly at random. Strategies for Bob and Charlie are then
determined by the state $\rho_{A_1 \ldots A_n BC}$ (with each $A_{j}$ being $d$-dimensional) as
well as independent measurements on $B$ and $C$ that produce 
a guess of the string $x$, for each value of $\theta = \theta_1 \ldots \theta_n \in \Theta^{n}$.
{In the following, we say that a game satisfies \emph{parallel repetition} if $\pwin(\mathsf{G}^{\times n})$ drops exponentially in $n$. Moreover, we say that it satisfies \emph{strong parallel repetition}
if this exponential drop is maximally fast, i.e. if $\pwin(\mathsf{G}^{\times n}) = \pwin(\mathsf{G})^n$.}

Returning to our example, Bob and Charlie could repeat the strategy that is optimal for a single round $n$ times to achieve a winning probability of $\pwin(\mathsf{G}_{\BB})^n = (1/2 + 1/(2\sqrt{2})^n \leq \pwin(\mathsf{G}_{\BB}^{\times n})$, but is this really the best they can do?
Even classically, analyzing the $n$-fold parallel repetition of games or tasks is typically challenging. Examples include the parallel repetition of interactive proof 
systems (see e.g.~\cite{raz98, hol07}) 
or the analysis of communication complexity tasks (see e.g.~\cite{klauck}). In a quantum world, such an analysis is often exacerbated further by the presence
of entanglement and the fact that quantum information cannot generally be copied. Famous examples include the analysis of the ``parallel repetition'' of channels
in quantum information theory (where the problem is referred to as the additivity of capacities) (see e.g.~\cite{matt:cap, graeme:qcap}), entangled non-local 
games~\cite{thomas:parallel}, 
or the question whether an eavesdropper's optimal strategy in quantum key distribution (QKD) is to perform the optimal strategy for each round. 
Fortunately, it turns out that strong parallel repetition does hold for our monogamy game. 

\begin{itemize}
	\item {\bf Main Result (informal):} We find 
	$\pwin(\mathsf{G}_{\BB}^{\times n}) = (1/2 + 1/(2\sqrt{2}))^n$.
	More generally, for all monogamy-of-entanglement games using incompatible measurements, we find that
	$\pwin(\mathsf{G}^{\times n})$ decreases exponentially in $n$. This also holds in the approximate case where Bob and Charlie are allowed to make a small fraction of errors.
\end{itemize}

Our proofs are appealing in their simplicity and use only tools from linear algebra, inspired by techniques proposed by Kittaneh~\cite{kittaneh97}.
Note that, in the more general case, we obtain parallel repetition, albeit not strong parallel repetition.

\subsection{Applications}

\subsubsection*{One-Sided Device Independent Quantum Key Distribution}

Quantum key distribution (QKD) makes use of quantum mechanical effects to allow 
two parties, Alice and Bob, to exchange a secret key while being eavesdropped by an attacker Eve~\cite{bb84,ekert91}. In principle, the security of QKD can be rigorously proven based solely on the laws of quantum mechanics~\cite{mayers96,SP00,renner05}; in particular, the security does not rely on the assumed hardness of some computational problem.
However, these security proofs typically make stringent assumptions about the devices used by Alice and Bob to prepare and measure the quantum states that are communicated. 
These assumptions are not necessarily satisfied by real-world devices, leaving the implementations of QKD schemes open to hacking attacks~\cite{lydersen10}. 

One way to counter this problem is by protecting the devices in an ad-hoc manner against known attacks. This is somewhat unsatisfactory in that the implementation may still be vulnerable to {\em unknown} attacks, and the fact that the scheme is in principle provably secure loses a lot of its significance. 

Another approach is to try to remove the assumptions on the devices necessary for the security proof; this leads to the notion of {\em device-independent} (DI) QKD. This line of research can be traced back to Mayers and Yao~\cite{my98} (see also~\cite{barrett05,acin07}). After some limited results (see, e.g.,~\cite{masanes11,haenggirenner10}), the possibility of DI QKD has recently been shown in the most general case by Reichhardt {\em et al}. in~\cite{reichardt12} and by Vazirani and Vidick in~\cite{vidick12}. 
In a typical DI QKD scheme, Alice and Bob check if the classical data obtained from the quantum communication violates a Bell inequality, which in turn ensures that there is some amount of fresh randomness in the data that cannot be known by Eve. This can then be transformed into a secret key using standard cryptographic techniques like information reconciliation and randomness extraction.  

While this argument shows that DI QKD is theoretically possible,
the disadvantage of such schemes 
is that they require a \emph{long-distance detection-loophole-free} violation of a Bell inequality by Alice and Bob.
This makes fully DI QKD schemes very hard to implement and very sensitive to any kind of noise and to inefficiencies of the physical devices: any deficiency will result in a lower observed (loophole free) Bell inequality violation, and currently conceivable experimental parameters are insufficient to provide provable security. 
Trying to find ways around this problem is an active line of research, 
see e.g.~\cite{gisin10,lo12,braunstein12,lim12,pironio12}.

Here, we follow a somewhat different approach, 
not relying on Bell tests, but making use of the {\em monogamy of entanglement}. Informally, the latter states that if Alice's state is fully entangled with Bob's, then it cannot be entangled with Eve's, and vice versa. As a consequence, if Alice measures a quantum system by randomly choosing one of two incompatible measurements, it is impossible for Bob and Eve to {\em both} have low entropy about Alice's measurement outcome. 
Thus, if one can verify that Bob has low entropy about Alice's measurement during the run of the scheme, it is guaranteed that Eve's entropy is high, and thus that a secret key can be distilled. 

Based on this idea, we show that the standard BB84 QKD scheme~\cite{bb84} is {\em one-sided} DI. 
This means that only Alice's quantum device has to be trusted, but no assumption about Bob's measurement device has to be made in order to prove security. Beyond that it does not communicate the measurement outcome to Eve, Bob's measurement device may be arbitrarily malicious. 


\begin{itemize}
\item \textbf{Application to QKD (informal):} We show that the BB84 QKD scheme is secure in the setting of fully one-sided device independence and
provide a complete security analysis for finite key lengths.
\end{itemize}

One-sided DI security of BB84 was first claimed in~\cite{tomamichel11}. However, a close inspection of their proof sketch, which is based on an entropic uncertainty relation with quantum side information, reveals that their arguments are insufficient to prove full one-sided DI security (as confirmed by the authors). It needs to be assumed that Bob's measurement device is {\em memoryless}. The same holds for the follow up work~\cite{tomamichel12,branciard12} of~\cite{tomamichel11}. 

Despite the practical motivation, our result is at this point of theoretical nature. This is because, as in all contemporary fully DI schemes, our analysis here (implicitly) assumes that every qubit sent by Alice is indeed received by Bob, or, more generally, whether it is received or not does not depend on the basis it is to be measured in; this is not necessarily satisfied in practical implementations\,---\,and some recent attacks on QKD take advantage of exactly this effect by blinding the detectors whenever a measurement in a basis not to Eve's liking is attempted~\cite{lydersen10}. {We remark here that this unwanted assumption can be removed in principle by a refined analysis along the lines of Branciard \emph{et al.}~\cite{branciard12}\footnote{There, the protocol of~\cite{tomamichellim11} was amended to account for photon losses.}. While this leads to a significantly lower key rate, the analysis in~\cite{branciard12} suggests that the loss tolerance for one-sided DI QKD is higher than for fully DI QKD. More precisely, while DI QKD requires a detection-loophole-free violation of a Bell inequality, for one-sided DI QKD a loophole-free violation of a steering inequality is sufficient, and such a violation has recently been shown~\cite{wittmann12}.}

Our analysis of BB84 QKD with one-sided DI security admits a noise level of up to $1.5\%$. This is significantly lower than the $11\%$ tolerable for standard (i.e.~not DI) security. We believe that this is not inherent to the scheme but an artifact of our analysis. Improving this bound by means of a better analysis is an open problem (it {\em can} be slightly improved by using a better scheme, e.g., the six-state scheme~\cite{bruss98}). Nonetheless, one-sided DI QKD appears to be an attractive alternative to DI QKD in an asymmetric setting, when we can expect from one party, say, a server, to invest in a very carefully designed, constructed, and tested apparatus, but not the other party, the user, and/or in case of a star network with one designated link being connected with many other links. 

A comparison 
to other recent results on device-independent QKD is given in Table~\ref{t1}. {The noise tolerance is determined using isotropic noise.}

\begin{table}[t]
\footnotesize
\begin{tabular}{c|c|c|c|c}
 & Reichhardt\,\emph{et\,al.}~\cite{reichardt12} & Vazirani/Vidick~\cite{vidick12} & this work & Tomamichel\,\emph{et\,al.}~\cite{tomamichellim11}\footnote{\scriptsize For comparison, this proof achieves maximum noise tolerance and key rate for BB84. See also~\cite{branciard12}.} \\ \hline
 protocol & E91-based~\cite{ekert91} & E91-based & BBM92~\cite{bennett92}/BB84~\cite{bb84} & asymmetric BB84~\cite{lo04}  \\
 device assumptions & none & none & trusted Alice\footnotemark\footnotetext[2]{\scriptsize Combining our results with results on self-testing in~\cite{haenggi11,lim12}, one can reduce the assumption to memoryless for Alice.}
& trusted Alice\footnotemark[2], \\
 & & & & memoryless Bob \\
 noise tolerance & 0\% & 1.2\% & 1.5\% & 11\% \\
 key rate (zero noise) & 0\% & 2.5\% & 22.8\%/11.4\%\footnote{\scriptsize This loss of a factor $\frac{1}{2}$ is due to sifting when moving from BBM92 to BB84.} & 1\\
 finite key analysis & no & no & yes & yes 
\end{tabular}
\caption{Comparison of Recent Fully and Partially Device-Independent Security Proofs for QKD.}
\label{t1}
\end{table}

\subsubsection*{Position Verification}

Our second application is to the task of {\em position verification}. Here, we consider a
$1$-dimensional setting where a {\em prover} wants to convince two {\em verifiers} that he controls a certain position, $\pos$. The verifiers are located at known positions around $\pos$, honest, and connected by secure communication channels. Moreover, all parties
are assumed to have synchronized clocks, and the message delivery time between any two parties
is assumed to be proportional to the distance between them. Finally, all local computations are assumed to be instantaneous.

Position verification and variants thereof (like {\em distance bounding}) is a rather well-studied problem in the field of wireless security (see e.g.~\cite{chandran09}).
It was shown in~\cite{chandran09} that in the presence of colluding adversaries at different locations, position verification 
is impossible classically, even with computational hardness assumptions. That is, the prover can always trick the verifiers into believing
that he controls a position.
The fact that the classical attack requires the adversary to {\em copy} information, initially gave  hope that we may circumvent the impossibility result 
using quantum communication~\cite{kent06,malaney10a,malaney10b,chandran10,kent10}. However,
such schemes were subsequently broken~\cite{lau11} and indeed a general impossibility proof holds~\cite{buhrman11}: 
without any restriction on the adversaries, in particular on the amount of pre-shared entanglement they may hold, no quantum scheme for position verification can be secure.
This impossibility proof was constructive but required the dishonest parties to share a number of EPR pairs that grows doubly-exponentially in the number of qubits the honest parties exchange. Using port-based teleportation, as introduced by Ishizaka and Hiroshima~\cite{ishizaka08,ishizaka09}, this was reduced by Beigi and K\"onig~\cite{beigi11} to a single exponential amount. On the other hand, there are schemes for position verification that are provably secure against adversaries that have no pre-shared entanglement, or only hold a couple of entangled qubits~\cite{chandran10,buhrman11,lau11,beigi11}.

However, all known schemes that are provably secure with a negligible soundness error (the maximal probability that a coalition of adversaries can pass the position verification test for position $\pos$ without actually controlling that specific position) against adversaries with no or with bounded pre-shared entanglement are either {\em multi-round} schemes, or require the honest participants to manipulate large quantum states.

\begin{itemize}
	\item \textbf{Application to Position Verification (informal):} We present 
the first provably secure \emph{one-round} position verification scheme with negligible soundness error in which the honest parties are only required to perform single qubit operations. 
We prove its security against adversaries with an amount of pre-shared entanglement that is 
\emph{linear} in the number of qubits transmitted by the honest parties.
\end{itemize}

\subsubsection*{Entropic Uncertainty Relation}

The final application of our monogamy game is to entropic uncertainty relations with quantum side information~\cite{berta10}.
Our result is in the spirit of~\cite{tomamichel11,coles11} which 
shows an uncertainty relation for a tripartite state $\rho_{ABC}$ for measurements on $A$, trading off the uncertainty between the two observers $B$ and $C$ as in our monogamy game.
\vspace{-0.2cm}
\begin{itemize}
	\item \textbf{Application to Entropic Uncertainty Relations:}
	For any two general (POVM) measurements, $\{ N_x^0 \}_x$ and $\{ N_x^1 \}_x$, we find
  \begin{align*}
    H_{\min}(X|B \Theta)_{\rho} + H_{\min}(X|C \Theta)_{\rho} \geq -2 \log \frac{1 + \sqrt{c}}{2} ,
     \quad \textrm{where} \quad c = \max_{x,z} \Big\| \sqrt{N_x^0}\sqrt{N_z^1} \Big\|^2 .
  \end{align*}
  The entropies are evaluated for the post-measurement state
    $\rho_{XBC\Theta}$, where $X$ is the outcome of the measurement
  $\{ N_x^\theta \}_x$, where $\Theta \in \01$ is chosen uniformly at random.
\end{itemize}
\vspace{-0.2cm}

\subsection{Outline}

The remainder of this manuscript is structured as follows. In Section~\ref{sc:pre}, we
introduce the basic terminology and notation used throughout this work. In Section~\ref{sec:GuessingGame}, we discuss the monogamy game and prove a strong parallel repetition theorem. Here, we also generalize the game to include the case where Bob and Charlie are allowed to have some errors in their guess and show
an upper bound on the winning probability for the generalized game. Sections~\ref{sc:qkd},~\ref{sc:pv} and~\ref{sc:ucr} then apply these results to prove security for one-sided device independent QKD, a one-round position verification scheme and an entropic uncertainty relation.

\section{Technical Preliminaries}\label{sc:pre}

\subsection{Basic Notation and Terminology}

Let $\sH$ be an arbitrary, finite dimensional Hilbert space.
$\olin{\sH}$ and $\opos{\sH}$ denote {\em linear} and {\em positive semi-definite} operators on $\sH$, respectively. Note that an operator $A \in \opos{\sH}$ is in particular {\em Hermitian}, meaning that $A^\dagger = A$. 
The set of {\em density operators} on $\sH$, i.e., the set of operators in $\opos{\sH}$ with unit trace, is denoted by $\dens(\sH)$.
For $A,B \in \olin{\sH}$, we write $A \geq B$ to express that $A - B \in \opos{\sH}$.
When operators are compared with scalars, we
implicitly assume that the scalars are multiplied by the identity operator, which we denote by~$\id_{\sH}$, or 
$\id$ if $\sH$ is clear from the context. 
A {\em projector}
is an
operator $P \in \opos{\sH}$ that satisfies $P^2 = P$.
A {\em POVM} (short for {\em positive operator valued measure}) is a set $\{ N_x \}_x$ of operators $N_x \in \opos{\sH}$ such that $\sum_x N_x = \id$, and 
a POVM is called \emph{projective} if all its elements $N_x$ are projectors.
We use the {\em trace distance} 
$$\Delta(\rho,\sigma) := \max_{0 \leq E \leq 1} \tr{E (\rho-\sigma)} = \frac{1}{2} \Tr \abs{\rho-\sigma}, \quad \textrm{where}\ \abs{L} = \sqrt{L^\dagger L}, $$ 
as a metric on density operators $\rho,\sigma \in \dens(\sH)$.

The most prominent example of a Hilbert space is the qubit, $\sH \equiv \bC^2$.
The vectors $\ket{0}$ and $\ket{1}$ form its {\em rectilinear} (or computational) basis, and the vectors 
$H\ket{0} = (\ket{0}+\ket{1})/\sqrt{2}$ and $H\ket{1} = (\ket{0}-\ket{1})/\sqrt{2}$ form 
its {\em diagonal} (or Hadamard) basis, where $H$ denotes the Hadamard matrix.
More generally, we often consider systems composed of $n$ qubits, $\sH \equiv \bC^2 \otimes \cdots \otimes \bC^2$.
For $x, \theta \in \{0,1\}^n$, we write $\ket{x^\theta}$ as a shorthand for the state vector 
$H^{\theta_1}\ket{x_1} \otimes \cdots \otimes H^{\theta_n}\ket{x_n} \in \sH$.

\subsection{The Schatten $\infty$-Norm}

For $L \in \olin{\sH}$, we use the Schatten $\infty$-norm $\norm{L} := \norm{L}_{\infty} = s_1(L)$, which evaluates the largest singular
value of $L$. 
It is easy to verify that this norm satisfies $\norm{L}^2 = \norm{L^{\dagger}L} = \norm{LL^{\dagger}}$. Also,
for $A, B \in \opos{\sH}$, $\norm{A}$ coincides the largest eigenvalue of $A$, and $A \leq B$ implies $\norm{A} \leq \norm{B}$. Finally, for block-diagonal operators we have
$\norm{A \oplus B} = \max \{ \norm{A}, \norm{B} \}$. We will also need the following norm inequality.





\begin{lemma}
  \label{lm:norm-monotonicity}
Let $A, B, L \in \olin{\sH}$ such that $A^{\dagger} A \geq B^{\dagger} B$. Then, it holds that $\norm[\big]{A L} \geq \norm[\big]{B L}$.
\end{lemma}

\begin{proof}
	First, note that $A^{\dagger} A \geq B^{\dagger} B$ implies that $L^{\dagger} A^{\dagger} A L \geq L^{\dagger} B^{\dagger} B L$ holds for an arbitrary linear operator $L$. By taking the norm we arrive at $\norm[\big]{L^{\dagger} A^{\dagger} A L} \geq \norm[\big]{L^{\dagger} B^{\dagger} B L}$, which is equivalent to $\norm[\big]{A L} \geq \norm[\big]{B L}$.
\end{proof}

In particular, if $A, A', B, B' \in \opos{\sH}$ satisfy $A' \geq A$ and $B' \geq B$ then applying the lemma twice (to the square roots of these operators) gives $\norm[\big]{\sqrt{A'} \sqrt{B'}} \geq \norm[\big]{\sqrt{A'} \sqrt{B}} \geq \norm[\big]{\sqrt{A} \sqrt{B}}$. For projectors the square roots can be omitted.

One of our main tools is the following Lemma~\ref{lm:mykittaneh}, which bounds the Schatten norm of the sum of $n$ positive semi-definite operators by means of their pairwise products. We derive the bound using a construction due to Kittaneh~\cite{kittaneh97}, which was also used by Schaffner~\cite{schaffner07} to derive a similar, but less general, result. 

We call two permutations $\pi: [N] \to [N]$ and $\pi': [N] \to [N]$ of the set $[N] := \{1,\ldots,N\}$ {\em orthogonal} if $\pi(i) \neq \pi'(i)$ for all $i \in [N]$. There always exists a set of $N$ permutations of $[N]$ that are
mutually orthogonal (for instance the $N$ cyclic shifts). 
%
\begin{lemma}
  \label{lm:mykittaneh}
  Let $A_1, A_2, \ldots, A_N \in \opos{\sH}$, and let $\{ \pi^k \}_{k \in [N]}$ be a set of $N$ mutually orthogonal permutations of $[N]$. Then,
  \begin{align}
    \norm[\Bigg]{\sum_{i \in [N]} A_i} \leq \sum_{k \in [N]}\, \max_{i \in [N]}
    \norm[\Big]{\sqrt{A_{i\phantom{\pi^k(i)}\!\!\!\!\!\!\!\!\!\!\!\!\!}}\sqrt{A_{\pi^k(i)}}} \label{lm1}\,.
  \end{align}
\end{lemma}

\begin{proof}
We define $X = [X_{ij}]$ as the $N \times N$ block-matrix with blocks given by $X_{ij} = \delta_{j1} \sqrt{A_i}$. Then, the matrices $X^{\dagger}X$ and $X X^{\dagger}$ are easy to evaluate, namely, $(X^{\dagger}X)_{ij} = \delta_{i1}\delta_{j1} \sum_i A_i$, as well as $XX^{\dagger}$ and $(XX^{\dagger})_{ij} = \sqrt{A_i}\sqrt{A_j}$. We have
\begin{align}
  \norm[\Bigg]{\sum_{i \in [N]} A_i} = \norm[\big]{X^{\dagger}X} = \norm[\big]{XX^{\dagger}} \nonumber\,.
\end{align}
Next, we decompose $XX^{\dagger} = D_1 + D_2 + \ldots D_N$, where the matrices $D_k$ are defined by the permutations $\pi^k$, respectively, as $(D_k)_{ij} = \delta_{j,\pi^k(i)} \sqrt{A_i}\sqrt{A_j}$. Note that
the requirement that the permutations are mutually orthogonal ensures that $XX^{\dagger} = \sum_k D_k$.
Moreover, since the matrices $D_k$ are constructed such that they contain exactly one non-zero 
block in each row and column, they can be transformed into a block-diagonal matrix 
$$D_k' = \bigoplus_{i \in [N]} \sqrt{A_{i\phantom{\pi^k(i)}\!\!\!\!\!\!\!\!\!\!\!\!}}\sqrt{A_{\pi^k(i)}}$$ 
by a unitary rotation. Hence, using the triangle inequality and unitary invariance of the norm, we get
$\norm[\big]{\sum_{k} A_{k}} \leq \sum_k \norm[\big]{D_k} = \sum_k \norm[\big]{D_k'}$, which implies~\eqref{lm1} since
$\norm[\big]{\bigoplus_i L_i} = \max_i \big\{ \norm{L_i} \big\}$.
\end{proof}

\noindent
A special case of the above lemma states that $\norm[\big]{A_1 + A_2} \leq \max \{ \norm{A_1}, \norm{A_2} \} + \norm[\big]{\sqrt{A_1}\sqrt{A_2}}$.

\subsection{CQ-States, and Min-Entropy}

A state $\rho_{XB} \in \dens(\sH_X \otimes \sH_B)$ is called a {\em classical-quantum} (CQ) state with classical $X$ over $\cal X$, if it is of the form 
$$
\rho_{XB} = \sum_{x\in\cX} p_x \proj{x}_X \otimes \rho_B^x \, ,
$$
where $\{ \ket{x} \}_{x \in {\cal X}}$ is a fixed basis of $\sH_X$, $\{ p_x \}_{x \in {\cal X}}$ is a probability distribution,  and $\rho_B^x \in \dens(\sH_B)$. 
For such a state, $X$ can be understood as a random variable that is correlated with (potentially quantum) side information $B$. 

If $\lambda: {\cal X} \to \{0,1\}$ is a predicate on $\cal X$, then we denote by $\Pr_{\rho}[\lambda(X)]$ the probability of the \emph{event} $\lambda(X)$ under $\rho$; formally,
$\Pr_{\rho}[\lambda(X)] = \sum_x p_x\, \lambda(x)$. We also define the state $\rho_{XB|\lambda(X)}$, which
is the state of the $X$ and $B$ conditioned on the event $\lambda(X)$. Formally,
$$
\rho_{XB|\lambda(X)} = \frac{1}{\Pr_{\rho}[\lambda(X)]} \sum_x p_x \lambda(x) \proj{x}_X \otimes \rho_B^x \, .
$$


For a CQ-state $\rho_{XB} \in \dens(\sH_X \otimes \sH_B)$, the {\em min-entropy} of $X$ conditioned on $B$ \cite{renner05}
can be expressed in terms of the maximum probability that a measurement on $B$ yields the correct
value of $X$, i.e.\ the guessing probability. Formally, we define~\cite{koenig08}
\begin{align*}
  H_{\min}(X|B)_{\rho} := - \log \pg(X|B)_{\rho}, \quad \textrm{where}\quad \pg(X|B)_{\rho} :=
  \max_{ \{ N_x \}_x } \sum_x p_x\, \tr{\rho_B^x N_x}.
\end{align*}
Here, the optimization is taken over all POVMs $\{ N_x \}_x$ on $B$, and here and throughout this paper, $\log$ denotes the binary logarithm. 

In case of a CQ-state $\rho_{XB\Theta}$ with classical $X$, and with additional classical side information $\Theta$, 
we can write 
$\rho_{XB\Theta} = \sum_{\theta} p_{\theta}\, \proj{\theta} \otimes \rho_{XB}^{\theta}$ for CQ states $\rho_{XB}^{\theta}$.
The min-entropy of $X$ conditioned on $B$ and $\Theta$ then evaluates to
\begin{align}
  H_{\min}(X|B\Theta)_{\rho} = - \log \pg(X|B\Theta)_{\rho}, \quad \textrm{where} 
  \quad \pg(X|B\Theta)_{\rho} =
  \sum_\theta p_\theta\, \pg(X|B)_{\rho^{\theta}} \, .
  \label{eq:hmin-theta}
\end{align}
An intuitive explanation of the latter equality is that
the optimal strategy to guess $X$ simply chooses an optimal POVM on $B$ depending on the value of $\Theta$. 

An overview of the min-entropy and its properties can be found in~\cite{renner05,mythesis}; we merely point out the {\em chain rule} here: for a CQ-state $\rho_{XB\Theta}$ with classical $X$ and $Y$, where $Y$ is over an arbitrary set $\mathcal{Y}$ with cardinality $|\mathcal{Y}|$, it holds that $H_{\min}(X|BY)_{\rho} \geq H_{\min}(X|B)_{\rho} - \log |\mathcal{Y}|$.


\section{Parallel Repetition of Monogamy Games}\label{sec:GuessingGame}

In this section, we investigate and show strong parallel repetition for the game $\mathsf{G}_{\BB}$. Then, we generalize our analysis to allow arbitrary measurements for Alice and consider the situation where Bob and Charlie are allowed
to make some errors. But to start with, we need some formal definitions. 

\begin{definition}
A \emph{monogamy-of-entanglement game} $\mathsf{G}$ consists of a finite dimensional Hilbert space $\sH_A$ and a list of measurements $\mathcal{M}^\theta = \{ F_x^{\theta} \}_{x \in \cX}$ on a $\sH_A$, indexed by $\theta \in \Theta$, where $\cX$ and $\Theta$ are finite sets. 
\end{definition}
\noindent 
We typically use less bulky terminology and simply call $\mathsf{G}$ a {\em monogamy game}. 
Note that for any positive integer $n$, the $n$-fold {\em parallel repetition} of $\mathsf{G}$, denoted as $\mathsf{G}^{\times n}$ and naturally specified by $\sH_A^{\otimes n}$ and $\{ F_{x_1}^{\theta_1} \otimes \cdots \otimes F_{x_n}^{\theta_n} \}_{x_1,\ldots,x_n}$ for $\theta_1,\ldots,\theta_n \in \Theta$, is again a monogamy game. 


\begin{definition}
We define a \emph{strategy} ${\cal S}$ for a monogamy game $\mathsf{G}$ as a list 
\begin{align}
{\cal S} = \big\{ \rho_{ABC},\, P_x^{\theta} ,\, Q_x^{\theta} \big\}_{\theta \in \Theta,x \in \cX}\ , \label{eq:strat}
\end{align}
where $\rho_{ABC} \in \dens(\sH_{A} \otimes \sH_B \otimes \sH_C)$, and $\sH_B$ and $\sH_C$ are arbitrary finite dimensional Hilbert spaces. 
Furthermore, for all $\theta \in \Theta$,
$\{ P^\theta_x \}_{x \in \cX}$ and $\{ Q^\theta_x \}_{x \in \cX}$ are POVMs on $\sH_B$ and $\sH_C$, respectively. \\
A strategy is called  \emph{pure} if the state $\rho_{ABC}$ is pure and all the POVMs are projective. 
\end{definition}
If ${\cal S}$ is a strategy for game $\mathsf{G}$, then the $n$-fold parallel repetition of $\cal S$, which is naturally given, is a particular strategy for the parallel repetition $\mathsf{G}^{\times n}$; however, it is important to realize that there exist strategies for $\mathsf{G}^{\times n}$ that are not of this form. 
In general, a strategy ${\cal S}_n$ for $\mathsf{G}^{\times n}$ is given by an arbitrary state $\rho_{A_1 \ldots A_n BC} \in \dens(\sH_A^{\otimes n} \otimes \sH_B \otimes \sH_C)$ (with arbitrary $\sH_B$ and $\sH_C$) and by arbitrary POVM elements on $\sH_B$ and $\sH_C$, respectively, not necessarily in product form.

The winning probability for a game $\mathsf{G}$ and a fixed strategy ${\cal S}$, denoted by $\pwin(\mathsf{G}, {\cal S})$, is defined as the probability that the measurement outcomes of Alice, Bob and Charlie agree when Alice measures in the basis determined
by a randlomly chosen $\theta \in \Theta$ and Bob and Charlie apply their respective POVMs $\{ P_x^{\theta} \}_x$ and $\{ Q_x^{\theta} \}_x$. 
The optimal winning probability, $\pwin(\mathsf{G})$, maximizes the winning probability over all strategies. The following makes this formal. 

\begin{definition}
  The winning probability for a monogamy game 
  $\mathsf{G}$ and a strategy ${\cal S}$ is defined as
  \begin{align}
    \pwin(\mathsf{G}, {\cal S}) := \sum_{\theta \in \Theta} \frac{1}{\abs{\Theta}} \Tr \big(\Pi^\theta \rho_{A BC}\big),
    \nbox{ where }{8} \Pi^\theta := \sum_{x \in \mathcal{X}} F_{x}^{\theta}  \otimes P_x^\theta \otimes Q_x^\theta .
    \label{eq:winS}
  \end{align}     
  The optimal winning probability is
  \begin{align}
    \pwin(\mathsf{G}) := \sup_{{\cal S}}\ \pwin(\mathsf{G}, {\cal S}) ,
    \label{eq:win}
  \end{align}
  where the supremum is taken over all strategies ${\cal S}$ for $\mathsf{G}$.
\end{definition}

In fact, due to a standard purification argument and Neumark's dilation theorem, we can restrict the supremum to pure strategies (cf.\ Lemma~\ref{lm:pure} in Appendix~\ref{app:pure}).

\subsection{Strong Parallel Repetition for $\mathsf{G}_{\BB}$}

%

We are particularly interested in the game $\mathsf{G}_{\BB}$ and its parallel repetition $\mathsf{G}_{\BB}^{\times n}$. The latter is given by $\sH_A = (\bC^2)^{\otimes n}$ and the projectors $F_{x}^{\theta} = \proj{x^\theta} = H^{\theta_1}\proj{x_1}H^{\theta_1} \otimes \cdots \otimes H^{\theta_n}\proj{x_n}H^{\theta_n}$ for $\theta,x \in \01^n$. 
The following is our main result. 

\begin{theorem}
  \label{th:perfect}
  For any $n \in \bN$, $n \geq 1$, we have
  \begin{align}
    \pwin(\mathsf{G}_{\BB}^{\times n}) = \bigg( \frac{1}{2} + \frac{1}{2 \sqrt{2}} \bigg)^n \, . \label{thm}
  \end{align}
\end{theorem}

\begin{proof}
We first show that this guessing probability can be achieved. For $n = 1$, consider
the following strategy. Bob and Charlie prepare the state $\ket{\phi} := \cos \frac{\pi}{8} \ket{0} + \sin \frac{\pi}{8} \ket{1}$ and send it to Alice. Then, they guess that Alice measures outcome $0$, independent of $\theta$.
Formally, this is the strategy
${\cal S}_1 = \big\{ \proj{\phi}, P_x^\theta = \delta_{x0}, Q_x^\theta = \delta_{x0} \big\}$. The optimal winning probability is thus bounded by the winning probability of 
this strategy, 
\begin{align*}
  \pwin(\mathsf{G}_{\BB}) \geq \Big(\cos \frac{\pi}{8} \Big)^2 = \frac{1}{2} + \frac{1}{2\sqrt{2}} \, ,
\end{align*}
and the lower bound on $\pwin$ implied by Eq.~\eqref{thm} follows by repeating this simple strategy $n$ times.

To show that this simple strategy is optimal, let us now fix an arbitrary, pure strategy ${\cal S}_n 
= \{ \rho_{A_1 \ldots A_n BC}, P_x^{\theta}, 
Q_x^{\theta} \}$. From the definition of the norm, we have $\tr{M \rho_{ABC}} \leq \norm{M}$ for any $M \geq 0$. Using this and
 Lemma~\ref{lm:mykittaneh}, we find 
\begin{align}
\pwin(\mathsf{G}_{\BB}^{\times n}, \mathcal{S}_n) =
\sum_\theta \frac{1}{2^n} \Tr\big(\Pi^\theta \rho_{A_1 \ldots A_n BC}\big) \leq 
\frac{1}{2^{n}}  \norm[\Big]{\sum_{\theta} \Pi^\theta} \leq 
\frac{1}{2^{n}} \sum_k\, \max_{\theta} \norm[\big]{\Pi^\theta \Pi^{\pi^k(\theta)}} \label{b1},
\end{align}
where the optimal permutations $\pi^k$ are to be determined later.
Hence, the problem is reduced to bounding the norms
$\norm[\big]{\Pi^\theta \Pi^{\theta'}}$, where $\theta' = \pi^k(\theta)$. The trivial upper bound on these norms, $1$, leads
to $\pwin(\mathsf{G}_{\textrm{BB84}}^{\times n}, \mathcal{S}_n) \leq 1$. However, most of these norms are actually very small as we see below.

For fixed $\theta$ and $k$, we denote by $\cT$ the set of indices where $\theta$ and $\theta'$ differ, by $\cT^c$ its complement, and by $t$ the Hamming distance between $\theta$ and $\theta'$ (hence, $t = \abs{\cT}$).
We consider the projectors
\begin{align}
  \bar{P} = \sum_{x} \ketbra{x_{\cT}^\theta} \otimes \id_{\cT^c} \otimes P_x^\theta \otimes \id_C \nbox{and}{8}
  \bar{Q} = \sum_{x} \ketbra{x_{\cT}^{\theta'}} \otimes \id_{\cT^c} \otimes \id_B \otimes Q_x^{\theta'} \nonumber,
\end{align}
where $\ket{x^\theta_{\cT}}$ is $\ket{x^\theta}$ restricted to the systems corresponding to rounds
with index in $\cT$, and $\id_{\cT^c}$ is the identity on the remaining systems.

Since $\Pi^\theta \leq \bar{P}$ and $\Pi^{\theta'}\! \leq \bar{Q}$, we can bound
$\norm[\big]{\Pi^\theta \Pi^{\theta'}}^2 \leq \norm[\big]{\bar{P} \bar{Q}}^{2} = \norm[\big]{\bar{P} \bar{Q} \bar{P}}$ using Lemma~\ref{lm:norm-monotonicity}.
Moreover, it turns out that the operator $\bar{P} \bar{Q} \bar{P}$ has a particularly simple form, namely
\begin{align*}
  \bar{P} \bar{Q} \bar{P} &= \sum_{x,y,z} \ket{x_{\cT}^{\theta}}\! \braket{x_{\cT}^{\theta}}{y_{\cT}^{\theta'}}\! \braket{y_{\cT}^{\theta'}}{z_{\cT}^{\theta}}\!\bra{z_{\cT}^{\theta}} \otimes \id_{\cT^c} \otimes P_x^{\theta} P_z^{\theta} \otimes Q_y^{\theta'} \\
  &= \sum_{x,y} \abs{\braket{x_{\cT}^\theta}{y_{\cT}^{\theta'}}}^2\, \ketbra{x_{\cT}^{\theta}} \otimes \id_{\cT^c} \otimes P_x^{\theta} \otimes Q_y^{\theta'} \\
  &= 2^{-t}\ \sum_{x} \ketbra{x_{\cT}^\theta} \otimes \id_{\cT^c} \otimes P_x^{\theta} \otimes \id_{C} ,
\end{align*}
where we used that $P_x^\theta P_z^\theta = \delta_{xz} P_x^\theta$ and $\abs{\braket{x_{\cT}^\theta}{y_{\cT}^{\theta'}}}^2 = 2^{-t}$. The latter relation follows from the fact that the two bases are diagonal
to each other on each qubit with index in $\cT$.
From this follows directly that $\norm{\bar{P} \bar{Q} \bar{P}} = 2^{-t}$. Hence, we find
$\norm[\big]{\Pi^\theta \Pi^{\theta'}} \leq \sqrt{2^{-t}}$. Note that this bound is independent of the strategy and only depends on the Hamming distance 
between $\theta$ and~$\theta'$.

To minimize the upper bound in~\eqref{b1}, we should choose permutations $\pi^k$ that produce tuples $(\theta, \theta' = \pi^k(\theta))$ with the same Hamming distance
as this means that the maximization is over a uniform set of elements.
A complete mutually orthogonal set of permutations with this property is given by the bitwise XOR, $\pi^k(\theta) = \theta \oplus k$, where we interpret $k$ as an element of $\{0, 1\}^n$. Using this construction, we get exactly ${n \choose t}$ permutations that create pairs with Hamming distance $t$, and
the bound in Eq.~\eqref{b1} evaluates to
\begin{align*}
\pwin(\mathsf{G}_{\BB}^{\times n}, \mathcal{S}_n) 
\leq \frac{1}{2^{n}} \sum_k\, \max_{\theta} \norm[\big]{\Pi^\theta \Pi^{\pi^k(\theta)}} 
\leq \frac{1}{2^{n}} \sum_{t=0}^n {n \choose t} \Big( \frac{1}{\sqrt{2}} \Big)^t  = \bigg( \frac12 + \frac{1}{2\sqrt{2}} \bigg)^n \, .
\end{align*}
Since this bound applies to all pure strategies, Lemma~\ref{lm:pure} concludes the proof.
\end{proof}

\subsection{Arbitrary Games, and Imperfect Guessing}
\label{sec:fixedError}

The above upper-bound techniques can be generalized to an arbitrary monogamy game, $\mathsf{G}$, specified by an arbitrary finite dimensional Hilbert space $\sH_A$ and arbitrary measurements $\{ F_x^{\theta} \}_{x \in \cal X}$, indexed by $\theta \in \Theta$, and with arbitrary finite $\cal X$ and $\Theta$.
The only additional
parameter relevant for the analysis is the {\em maximal overlap} of the measurements,
\begin{align}
  c(\mathsf{G}) := \max_{\theta,\theta' \in \Theta \atop \theta \neq \theta'} 
  \max_{x,\,x' \in \cX} \norm[\Big]{ \sqrt{F_x^\theta} \sqrt{F_{x'}^{\theta'}}}^2, \nonumber\,
\end{align}
which satisfies $1/\abs{\cX} \leq c(\mathsf{G}) \leq 1$ and $c(\mathsf{G}^{\times n}) = c(\mathsf{G})^n$. 
This is in accordance with the definition of the overlap as it appears in entropic uncertainty relations,
e.g.\ in~\cite{krishna01}. Note also that in the case of $\mathsf{G}_{\BB}$, we have
$c(\mathsf{G}_{\BB}) = \frac{1}{2}$.

In addition to considering arbitrary monogamy games, we also generalize Theorem~\ref{th:perfect} to the case where Bob and Charlie are not required to guess the outcomes {\em perfectly} but are allowed to make some errors. 
The maximal winning probability in this case is defined as follows, where we employ an argument analogous to Lemma~\ref{lm:pure} in order to restrict to pure strategies. 

\begin{definition}
Let $\cQ = \{ ( \pi_B^q, \pi_C^q ) \}_q$ be a set of pairs of
  permutations of $\cX$, indexed by $q$, with the meaning that in order to win, Bob and Charlie's respective guesses for $x$ must form a pair in $\{(\pi_B^q(x), \pi_C^q(x) ) \}_q$. 
  Then, the optimal winning probability of $\mathsf{G}$ with respect to $\cQ$ is
$$
\pwin(\mathsf{G}; \cQ) := \sup_{{\cal S}}\
    \sum_{\theta \in \Theta} \frac{1}{\abs{\Theta}} \tr{A^\theta \rho_{A BC}}
\quad\text{with}\quad
    A^\theta := \sum_{x \in \mathcal{X}}  F_{x}^{\theta} 
      \otimes \sum_{q} P_{\pi_B^q(x)}^\theta \otimes Q_{\pi_C^q(x)}^\theta   ,
$$
where the supremum is taken over all pure strategies ${\cal S}$ for $\mathsf{G}$. 
\end{definition}

We find the following upper bound on the guessing probability, generalizing the upper bound
on the optimal winning probability established in Theorem~\ref{th:perfect}.

\begin{theorem}
  \label{th:gen}
  For any positive $n \in \mathbb{N}$, we have
  \begin{align}
    \pwin(\mathsf{G}^{\times n}; \cQ) \leq \abs{\cQ}  \bigg( \frac{1}{\abs{\Theta}} + \frac{\abs{\Theta} - 1}{\abs{\Theta}}\, \sqrt{c(\mathsf{G})} \bigg)^n \nonumber\,.
  \end{align}
\end{theorem}
Recall that in case of $\mathsf{G}_{\BB}$, we have $\abs{\cQ} = 1$, $\abs{\Theta} = 2$, and $c(\mathsf{G}_{\BB}) = \frac{1}{2}$, leading to the bound stated in Theorem~\ref{th:perfect}.

\begin{proof}
We closely follow the proof of the upper bound in Theorem~\ref{th:perfect}. 
For any pure strategy ${\cal S}_n = \{ \rho_{A_1 \ldots A_n BC}, P_x^{\theta}, Q_x^{\theta}\}$, we bound
\begin{align}
\sum_\theta \frac{1}{\abs{\Theta}^n} \tr{A^\theta \rho_{A_1 \ldots A_n BC}}
\leq \frac{1}{\abs{\Theta}^{n}}  
    \norm[\Big]{\sum_{\theta} A^\theta} \leq 
\frac{1}{\abs{\Theta}^{n}} \sum_q \sum_k \max_{\theta}
    \norm[\bigg]{\sqrt{A_q^{\phantom{\pi^k()\hspace{-0.5cm}}\theta}} \sqrt{A_q^{\pi^k(\theta)}}} \label{blah5},
\end{align}
where we introduce $A_q^\theta := \sum_{x} \big( \bigotimes_{\ell=1}^n
F_{x_\ell}^{\theta_\ell} \big) \otimes P_{\pi_B^q(x)}^\theta \otimes Q_{\pi_C^q(x)}^\theta$.
We now fix $\theta$ and $\theta'$ and bound the norms $\norm[\Big]{\sqrt{A_q^\theta} \sqrt{A_q^{\theta'}}}$.
Let $\cT$ be the set of indices where $\theta$ and $\theta'$ differ. We choose
\begin{align}
  B = \sum_{x} \bigotimes_{\ell \in \cT}  F_{x_{\ell}}^{\theta_{\ell}}
  \otimes \id_{\cT^c} \otimes P_{\pi_B^q(x)}^{\theta} \otimes \id_{C} \ \nbox{and}{2} \
  C = \sum_{x} \bigotimes_{\ell \in \cT}  F_{x_{\ell}}^{\theta_{\ell}'} \otimes
  \id_{\cT^c} \otimes \id_{B} \otimes Q_{\pi_C^q(x)}^{\theta'},  \nonumber
\end{align}
which satisfy $B \geq A_{q}^{\theta}$ and $C \geq A_{q}^{\theta'}$. Hence, from Lemma~\ref{lm:norm-monotonicity} we obtain $\norm[\Big]{\sqrt{A_q^\theta} \sqrt{A_q^{\theta'}}} \leq \norm[\big]{\sqrt{B}\sqrt{C}}$. We evaluate
\begin{align}
  \norm[\big]{\sqrt{B} \sqrt{C}} = \norm[\Bigg]{\sum_{x, y}
    \bigotimes_{\ell \in \cT} \sqrt{F_{x_{\ell}}^{\theta_{\ell}}} \sqrt{F_{y_{\ell}}^{\theta_{\ell}'}}
  \otimes \id_{\cT^c} \otimes P_{\pi_B^q(x)}^{\theta} \otimes Q_{\pi_C^q(y)}^{\theta'} } \nonumber
    = \max_{x, y} \norm[\bigg]{ \bigotimes_{\ell \in \cT} \sqrt{F_{x_{\ell}}^{\theta_{\ell}}} \sqrt{F_{y_{\ell}}^{\theta_{\ell}'}}} \leq c(\mathsf{G})^t \nonumber.
\end{align}

It remains to find suitable permutations $\pi^k$ and substitute the above bound into~\eqref{blah5}. Again, we choose permutations with the property that the Hamming distance between $\theta$ and $\pi^{k}(\theta)$ is the same for all $\theta \in \Theta^{n}$. It is easy to verify that there are ${n \choose t} \big( \abs{\Theta} - 1)^t$ permutations for which the ($\theta$-independent) Hamming distance between $\theta$ and $\pi^{k}(\theta)$ is $t$. Hence,
\begin{align}
\sum_\theta \frac{1}{\abs{\Theta}^n} \tr{\Pi^\theta \rho_{A_1 \ldots A_n BC}}
\leq \frac{\abs{\cQ}}{\abs{\Theta}^{n}} \sum_{t = 0}^{n} {n \choose t}
    \big( \abs{\Theta} - 1 \big)^t  (\sqrt{c(\mathsf{G})})^t = \abs{\cQ} \bigg( \frac{1}{\abs{\Theta}} + \frac{\abs{\Theta} - 1}{\abs{\Theta}}\, \sqrt{c(\mathsf{G})}
    \bigg)^n \nonumber \, ,
\end{align}
which concludes the proof. 
\end{proof}

One particularly interesting example of the above theorem considers binary measurements, i.e.\ $\mathcal{X} = \{0,1\}$, where
Alice will accept Bob's and Charlie's answers if and only if they get less than a certain fraction of bits wrong.
More precisely, she accepts if $d(x, y) \leq \gamma\, n$ and $d(x,z) \leq \gamma'\,n$,
where $d(\cdot,\cdot)$ denotes the Hamming distance and $y$, $z$ are Bob's and Charlie's guesses, respectively. In this case, we introduce the set $\cQ_{\gamma,\gamma'}^n$ that contains all pairs
of permutations $(\pi_B^q, \pi_C^q)$ on $\{0,1\}^n$ of the form
$\pi_B^q(x) = x \oplus k$, $\pi_C^q(x) = x \oplus k'$, where $q = \{k, k'\}$, and $k, k' \in \{0,1\}^n$ have Hamming weight
at most $\gamma n$ and $\gamma' n$, respectively. 
For $\gamma, \gamma' \leq 1/2$, one can upper bound
$\abs{\cQ_{\gamma,\gamma'}^n} \leq 2^{n h(\gamma) + n h(\gamma')}$, where $h(\cdot)$ denotes the binary entropy. We thus find
\begin{align}
  \pwin(\mathsf{G}^{\times n}; \cQ_{\gamma,\gamma'}^n) \leq \bigg( 2^{h(\gamma) + h(\gamma')}\, \frac{1 + (\abs{\Theta} - 1) \sqrt{c(\mathsf{G})}}{\abs{\Theta}} \bigg)^n \label{eq:error}.
\end{align}
Similarly, if we additionally require that Charlie guesses the same string as Bob, we analogously 
define the corresponding set $\cQ_{\gamma}^n$, with reduced cardinality, and
\begin{align}
  \pwin(\mathsf{G}^{\times n}; \cQ_\gamma^n) \leq \bigg( 2^{h(\gamma)}\, \frac{1 + (\abs{\Theta} - 1) \sqrt{c(\mathsf{G})}}{\abs{\Theta}}
  \bigg)^n \nonumber.
\end{align}
%

\section{Application: One-Sided Device-Independent QKD}\label{sc:qkd}

\newcommand{\drel}{d_{\mathrm{rel}}}
\newcommand{\syn}{\mathrm{syn}}

In the following, we assume some familiarity with quantum key distribution (QKD). 
%
For simplicity, we consider an entanglement-based~\cite{ekert91} variant of the BB84 QKD scheme~\cite{bb84}, where Bob waits with performing the measurement until Alice tells him the right bases. This protocol is impractical because it requires Bob to store qubits. However, it is well known that security of this impractical version implies security of the original, more practical BB84 QKD scheme~\cite{bennett92}. It is straightforward to verify that this implication also holds in the one-sided device-independent setting we consider here. 

The entanglement-based QKD scheme, {\bf E-QKD}, is described in Figure~\ref{fig:EPRQKD}. It is (implicitly) parameterized by positive integers $0 < t,s,\ell < n$ and a real number $0 \leq \gamma < \frac12$. Here, $n$ is the number
of qubits exchanged between Alice and Bob, $t$ is the size of 
the sample used for parameter estimation, $s$ is the leakage (in bits) due to error correction, $\ell$ is the length (in bits) of the final key, and $\gamma$ is the tolerated error in Bob's measurement
results. 
Furthermore, the scheme makes use of a universal$_2$ family $\cal F$ of hash functions $F: \{0,1\}^{n-t} \to \{0,1\}^{\ell}$. 


\begin{figure}
\begin{framed}
\begin{description}\setlength{\parskip}{0.5ex}
  \item[\sl State Preparation: ] Alice prepares $n$ EPR pairs 
    $\frac{1}{\sqrt{2}} \big( \ket{0} \otimes \ket{0} + \ket{1} \otimes \ket{1} \big)$.  Then, of each pair, she keeps one qubit and sends the other to Bob. 
  \item[\sl Confirmation: ] Bob confirms receipt of the $n$ qubits. (After this point, there cannot be any communication
  between Bob's device and Eve.)
  \item[\sl Measurement: ] Alice chooses random $\Theta \in \{0,1\}^n$ and sends it to Bob, and Alice and Bob measure the EPR pairs in basis 
    $\Theta$ to obtain $X$ and $Y$,   respectively. \\
    (Remember: Bob's device may produce $Y$ in an arbitrary way, using a POVM chosen depending on $\Theta$ acting
    on a state provided by Eve.)
  \item[\sl Parameter Estimation: ] Alice chooses a random subset $T \subset \{1,\ldots,n\}$ of size $t$, and sends $T$ and $X_T$ to Bob.
    If the relative Hamming distance,
    $\drel(X_T,Y_T)$, exceeds $\gamma$ then they abort the protocol and set $K =\ \perp$.
  \item[\sl Error Correction: ] Alice sends a syndrome $S(X_{\bar{T}})$ of length $s$
    and a random hash function $F: \{0,1\}^{n-t} \to \{0,1\}^{\ell}$ from $\cal F$ to Bob.
  \item[\sl Privacy Amplification: ] Alice computes $K = F(X_{T^c})$ and Bob $\hat{K} = F(\hat X_{T^c})$, 
    where $\hat X_{T^c}$ is the corrected version of $Y_{T^c}$. 
\end{description}\vspace{-3ex}
\end{framed}\vspace{-2ex}
\caption{An entanglement-based QKD scheme~{\bf E-QKD}. }\label{fig:EPRQKD}
\end{figure}

A QKD protocol is called {\em perfectly secure} if it either aborts and outputs an empty key, 
$K =\, \perp$, or it produces a key that is uniformly random and independent of Eve's (quantum and classical) information $E^+$ gathered during the execution of the protocol. Formally, this means that the final state must be of the form 
$\rho_{KE^+} = \Pr_{\rho}[K\neq\,\perp] \cdot \mu_K \otimes \rho_{E^+|K\neq\perp} + \Pr_{\rho}[K =\,\perp] \cdot |\!\!\perp\rangle\!\langle\perp\!\!|_K \otimes \rho_{E^+|K=\perp}$, 
where $\mu_K$ is a $2^\ell$-dimensional completely mixed state, and $|\!\!\perp\rangle\!\langle\perp\!\!|_K$ is 
orthogonal to~$\mu_K$.

Relaxing this condition, a protocol is called {\em $\delta$-secure} if $\rho_{KE^+}$ is $\delta$-close to the above form
in trace distance, meaning that $\rho_{KE^+}$ satisfies 
\begin{align}
  \Pr_{\rho}[K \neq \, \perp] \cdot \Delta( \rho_{KE^+|K\neq\perp}, \mu_K \otimes \rho_{E^+|K\neq\perp}) \leq \delta \,.
\end{align}

It is well known and has been proven in various ways that {\bf E-QKD} is $\delta$-secure (with small~$\delta$) with a suitable choice of parameters, assuming that all quantum operations are correctly performed by Alice and Bob. We now show that the protocol remains secure even if Bob's measurement device behaves arbitrarily and possibly maliciously. The only assumption is that Bob's device does not
communicate with Eve after it received Alice's quantum signals. This restriction is clearly necessary as there would otherwise not be any asymmetry between Bob and Eve's information about Alice's key. Note that the scheme is well known to satisfy {\em correctness} and {\em robustness}; hence, we do not argue these here. 


\begin{theorem}
Consider an execution of {\bf E-QKD}, with an arbitrary measurement device for Bob. Then, 
for any $\eps > 0$, protocol 
{\bf E-QKD} is $\delta$-secure with
\begin{align*}
\delta = 5 e^{-2 \eps^2 t} + \ 2^{-\frac12\big(\log(1/\beta_\circ)n-h(\gamma+\eps)n-\ell-t-s+2\big)}
\quad
\textrm{where} \quad \beta_\circ = \frac12+\frac{1}{2\sqrt{2}}. 
\end{align*}
\end{theorem}

%
Note that with an optimal error correcting code, the size of the syndrome for large $n$ approaches the
Shannon limit $s = n h(\gamma)$. The security error $\delta$ can then be made negligible in $n$ with suitable choices of parameters if $\log(1/\beta_\circ) > 2 h(\gamma)$, which roughly requires that $\gamma \leq 0.015$. 
Hence, the scheme can tolerate a noise level up to $1.5\%$ asymptotically.\footnote{This can be improved slightly by instead considering a six-state protocol~\cite{bruss98}, where the measurement is randomly chosen among three mutually unbiased bases on the qubit.}

The formal proof is given below. The idea is rather simple: We consider a \emph{gedankenexperiment} where Eve {\em measures} her system, using an arbitrary POVM, with the goal to guess~$X$. The execution of {\bf E-QKD} then pretty much coincides with $\mathsf{G}_{\BB}^{\times n}$, and we can conclude from our results that if Bob's measurement outcome $Y$ is close to $X$, then Eve must have a hard time in guessing $X$. Since this holds for any measurement she may perform, this means her min-entropy on $X$ is large 
and hence the extracted key $K$ is secure.

\begin{proof}
Let $\rho_{\Theta T ABE} = \rho_\Theta \otimes \rho_T \otimes \ketbra{\psi_{ABE}}$ be the state before Alice and Bob perform the measurements on $A$ and $B$, respectively, where system $E$ is held by the adversary Eve. Here, the random variable $\Theta$ contains the choice of basis for the measurement, whereas the random variable $T$ contains the choice
of subset on which the strings are compared (see the protocol description in Fig.~\ref{fig:EPRQKD}.)
Moreover, let $\rho_{\Theta T XYE}$ be the state after Alice and Bob measured, where\,---\,for every possible value $\theta$\,---\,Alice's measurement is given by the projectors 
$\{\ketbra{x^\theta}\}_x$, and Bob's measurement by an arbitrary but fixed POVM $\{ P^\theta_x \}_x$. 

As a \emph{gedankenexperiment}, we consider the scenario where Eve wants to guess the value of Alice's
raw key, $X$. Eve wants to do this during the parameter estimation step of the protocol, exactly 
\emph{after} Alice broadcast $T$ but \emph{before} she broadcasts $X_T$.\footnote{Note that the effect
of Eve learning $X_T$ is taken into account later, in Eq.~\eqref{eq:penalty}.}
For this purpose, we consider an arbitrary measurement strategy of Eve that aims to guess 
$X$. Such a strategy
is given by\,---\,for every basis choice, $\theta$, and every choice of sample, $\tau$\,---\,a POVM $\{ Q^{\theta,\tau}_x \}_x$. The values of $\theta$ and $\tau$ have been broadcast over a public channel, and are hence known to Eve at this point of the protocol. She will thus choose a POVM depending on these values to measure $E$ and use the measurement outcome as her guess.

For our \emph{gedankenexperiment}, we will use the state, 
$\rho_{\Theta T XYZ}$, which is the (purely classical) state that results after Eve applied her measurement on $E$. 
Let $\varepsilon >0 $ be an arbitrary constant. By our results from Section~\ref{sec:GuessingGame}, it follows that for any choices of $\{ P^\theta_x \}_x$ and $\{ Q^{\theta,\tau}_x \}_x$, we have 
$$\Pr_{\rho}[ \drel(X,Y) \!\leq\! \gamma\!+\!\eps \,\wedge\, Z \!=\! X ] \leq \pwin(\mathsf{G}_{\BB}^{\times n}; \cQ_{\gamma+\varepsilon,0}^n) \leq \beta^n 
$$
with $\beta = 2^{h(\gamma+\eps)} \cdot \beta_\circ$, 
where $\drel$ denotes the relative Hamming distance. This uses the fact that Alice's measurement outcome is independent of $T$, and $T$ can in fact be seen as part of Eve's system for the purpose of the monogamy game.


We now construct a state $\tilde{\rho}_{\Theta T XYE}$ as follows.
$$
\tilde{\rho}_{\Theta T XYE} =  \Pr_{\rho}[ \Omega] \cdot \rho_{\Theta T XYE|\Omega} +  \big( 1 - \Pr_{\rho}[\Omega] \big) \cdot \sigma_{\Theta T XYE} ,
$$
where $\Omega$ denotes the event $\Omega = \{ \drel(X,Y) \leq \drel(X_T,Y_T) + \varepsilon \}$, and we take $\sigma_{T \Theta XYE}$ to be an arbitrary state with classical 
$\Theta$, $T$, $X$ and $Y$ for which $\drel(X,Y) = 1$, and hence $\drel(X_T,Y_T) = 1$. Informally, the event $\Omega$ indicates that the relative Hamming distance of the sample strings $X_T$ and $Y_T$ 
determined by $T$ was representative of the relative Hamming distance between the 
whole strings, $X$ and $Y$, and the state $\tilde{\rho}_{\Theta T XYE}$ is so that this is satisfied with certainty. 
By construction of $ \tilde{\rho}_{\Theta T XYE}$, we have $\Delta(\rho_{\Theta T XYE}, \tilde{\rho}_{\Theta T XYE}) 
\leq 1 - \Pr_{\rho}[\Omega]$, and by Hoeffding's inequality, 
\begin{align*}
  1 - \Pr_{\rho}[\Omega] = \Pr_{\rho}[ \drel(X,Y) > \drel(X_T,Y_T) + \varepsilon ] \leq e^{-2 \eps^2 t} . 
\end{align*}
Moreover, note that the event $\drel(X_T,Y_T) \leq \gamma$ implies $\drel(X,Y) \leq \gamma+\varepsilon$
under $\tilde{\rho}_{\Theta T XYE}$.
Thus, for every choice of strategy $\{ Q^{\theta,\tau}_x \}_x$ by the eavesdropper, the resulting state $\tilde{\rho}_{\Theta T XYZ}$, obtained by applying $\{ Q^{\theta,\tau}_x \}_x$ to $E$, satisfies
\begin{align}
\Pr_{\tilde{\rho}}[\drel(X_T,Y_T) \!\leq\! \gamma \wedge Z \! = \! X] &\leq \Pr_{\tilde{\rho}}[\drel(X,Y) \!\leq\! \gamma\!+\!\eps \wedge Z \! = \! X] \label{eq:b1}\\
&\leq \Pr_\rho[\drel(X,Y) \!\leq\! \gamma\!+\!\eps \wedge Z \!=\! X] \leq \beta^n . \nonumber
\end{align}
The second inequality follows from the definition of $\tilde{\rho}$, in particular the fact that $\Pr_{\sigma} [ \drel(X,Y) \leq \gamma + \eps ] = 0$.

Next, we introduce the event $\Gamma = \{ \drel(X_T, Y_T) \leq \gamma \}$, which corresponds to the event that Bob does not abort the protocol. 
Expanding the left hand side of~\eqref{eq:b1} to 
$\Pr_{\tilde{\rho}}[\Gamma] \cdot \Pr_{\tilde{\rho}}[Z \! = \! X | \Gamma]$
and observing that $\Pr_{\tilde{\rho}}[\Gamma]$ 
does not depend on the strategy $\{ Q^{\theta,\tau}_x \}_x$, we can conclude that 
\begin{align*}
  \forall \, \{ Q^{\theta,\tau}_x \}_x : \  \Pr_{\tilde{\rho}}[Z \! = \! X | \Gamma ] \leq \beta^{(1-\alpha)n} 
\end{align*}
where $\alpha \geq 0$ is determined by $\Pr_{\tilde{\rho}}[\Gamma ] = \beta^{\alpha n}$. 
Therefore, by definition of the min-entropy, $\Hmin(X | \Theta T E, \Gamma)_{\tilde{\rho}} \geq n (1\!-\!\alpha) \log(1/\beta)$.  
(This notation means that the min-entropy of $X$ given $\Theta$, $T$ and $E$ is evaluated for the state $\tilde{\rho}_{\Theta T X Y E | \Gamma}$, conditioned on not aborting.) 
By the chain rule, it now follows that 
\begin{align}\label{eq:penalty}
  \Hmin(X | \Theta T X_T S E, \Gamma)_{\tilde{\rho}} &\geq \Hmin(X X_T S | \Theta T E, \Gamma)_{\tilde{\rho}} - t - s \\
  &\geq n (1-\alpha) \log(1/\beta) - t - s \,. \nonumber
\end{align}
Here, the min-entropy is evaluated for the state $\tilde{\rho}_{X \Theta T X_T S E}$ that is constructed 
from $\tilde{\rho}_{X \Theta T E}$ by calculating the error syndrome and copying $X_T$ from $X$ as
done in the prescription of the protocol. In particular,  
$\Delta(\tilde{\rho}_{X \Theta T X_T S E}, \rho_{X \Theta T X_T S E}) \leq e^{-2\eps^2t}$.
Finally, privacy amplification with universal$_2$ hashing applied to the state $\tilde{\rho}_{X \Theta T X_T S E}$ ensures that the key $K$ satisfies~\cite[Corollary 5.5.2]{renner05}
\begin{align*}
\Delta(\tilde{\rho}_{K F \Theta T X_T S E | \Gamma},\, \mu_K \otimes \tilde{\rho}_{F \Theta T X_T E | \Gamma})
&\leq \frac{1}{2} \sqrt{ 
\beta^{(1-\alpha)n}\, 2^{\ell + t + s}} \, .
\end{align*}
And, in particular, recalling that $\Pr_{\tilde{\rho}}[\Gamma ] = \beta^{\alpha n}$, we have
\begin{align*}
  \Pr_{\tilde
{\rho}}[\Gamma ] \cdot \Delta(\tilde{\rho}_{K F \Theta T X_T S E | \Gamma},\, \mu_K \otimes \tilde{\rho}_{F \Theta T X_T E | \Gamma}) \leq \frac{1}{2} \sqrt{\beta^{n}\, 2^{\ell + t + s}} \, .
\end{align*}
Using $\beta = 2^{h(\gamma+\eps)} \beta_\circ$ and applying Lemma~\ref{lemma:SecDefs} in Appendix~\ref{app:secdefs} concludes the proof. 
\end{proof}

\section{Application II: A One-Round Position-Verification Scheme}\label{sc:pv}

The scheme we consider is the parallel repetition of the simple single-qubit scheme that was analyzed in the setting of no pre-shared entanglement in~\cite{buhrman11}. The analysis shows that the soundness error of the one-round single-qubit scheme is bounded by roughly $89\%$, and it is suggested to repeat the scheme sequentially in order to reduce this soundness error. We now show that also the {\em parallel repetition} has an exponentially small soundness error.%
\footnote{We stress that this was to be expected and does not come as a surprise. However, until now it was unclear how to prove it.}
Finally, we use a simple observation from~\cite{beigi11} to argue that the scheme is also secure against adversaries with a linearly bounded amount of entanglement.

The scheme, parameterized by a positive integer $n$, consists of the following steps.
\begin{enumerate}
\item $V_0$ and $V_1$ agree on random $x,\theta \in \{0,1\}^n$. $V_0$ prepares a quantum system $Q$ of $n$ qubits in the state $H^\theta \ket{x} = H^{\theta_1} \ket{x_1} \otimes \cdots \otimes H^{\theta_n} \ket{x_n} \in \sH_Q = (\bC^2)^{\otimes n}$ and sends it to $P$. $V_1$ sends $\theta$ to $P$, so that both arrive at $P$'s claimed position $\pos$ at the same time.
\item As soon as $Q$ and $\theta$ arrive, $P$ measures the $i$-th qubit in basis $\{H^{\theta_i} \ket{0}, H^{\theta_i} \ket{1}\}$ for $i = 1,\ldots,n$. Let $x' \in \{0,1\}^n$ collect the observed bits. $P$ sends $x'$ to $V_0$ and $V_1$.
\item If $V_0$ and $V_1$ receive $x'$ at the respective time consistent with $\pos$, and if $x' = x$, then $V_0$ and $V_1$ accept; otherwise, they reject.
\end{enumerate}
It is straightforward to verify that this protocol is correct, meaning that the verifiers accept honest $P$ at position $\pos$ with certainty (assuming a perfect setting with no noise, etc.).

\begin{proposition}\label{prop:posver}
The above position verification scheme is $(\frac{1}{2}\!+\!\frac{1}{2\sqrt{2}})^n$-sound against adversaries $(E_0,E_1)$ that hold no entangled state at the time they receive $Q$ and $\theta$, respectively.
\end{proposition}
We stress that a restriction on the entanglement is necessary, as with unbounded entanglement the general impossibility result from~\cite{buhrman11} applies. In fact, for the specific scheme considered here, already $n$ shared EPR-pairs are sufficient to break it, as shown in~\cite{kent10}.
Below, we will extend the security of the scheme to a setting where the adversaries share at most $\alpha n$ entangled qubits, for any constant $\alpha \lesssim 0.22$.

We also point out that our adversary model (with linearly bounded entanglement) is stronger than the one considered by Beigi and K\"onig~\cite{beigi11} for their schemes: their model not only prohibits quantum communication between the adversaries {\em before} they obtain the initial messages from the verifiers (in order to prevent the exchange of entangled states), but also {\em afterwards}. Here, we allow full quantum communication between the adversaries after they have received the initial respective messages $Q$ and $\theta$.

\begin{proof}[Proof (sketch)]
As the colluding dishonest parties $E_0$ and $E_1$ share no entanglement, the most general attack is of the following form, where we may assume $E_i$ to be located between $V_i$ and the position $\pos$, for $i \in \{0,1\}$. Upon receiving the $n$-qubit system $Q$ (in state $H^\theta \ket{x}$) from $V_0$, the adversary $E_0$ applies an isometry $\sH_Q \to \sH_B \otimes \sH_C$ to $Q$ in order to obtain a bipartite system $B$ and $C$, and forwards $C$ to~$E_1$. Adversary $E_1$, upon receiving $\theta$ from $V_1$, simply forwards $\theta$ to $E_0$.\footnote{This is where the restriction of no entanglement comes into play. If the adversaries shared entanglement their most general strategy would be to perform some joint operation on the respective part of the entangled state and the data they have just received. The impossibility result states that in a scenario with an unlimited amount of entanglement no position verification scheme can be secure.} Then, when $E_0$ receives $\theta$ from $E_1$, he measures $B$ (using an arbitrary measurement that may depend on $\theta$) and sends the measurement outcome $x'_0 \in \{0,1\}^n$ to $V_0$, and, similarly, when $E_1$ receives system $C$ from $E_0$, he measures $C$ and sends the measurement outcome $x'_1 \in \{0,1\}^n$ to $V_1$. The probability $\varepsilon$ that $V_0$ and $V_1$ accept is then given by the probability that $x'_0 = x = x'_1$.

From a standard purification argument it follows that the probability $\varepsilon$ does not change if in the first step of the protocol, instead of sending $Q$ in state $H^\theta \ket{x}$, $V_0$ prepares $n$ EPR pairs, sends one half of each pair towards $P$ and only at some later point in time measures the remaining $n$ qubits in the basis $\{H^\theta\ket{y}\}_{y \in \{0,1\}^n}$ to obtain $x \in \{0,1\}^n$.

Let us now consider the state $\ket{\psi_{ABC}} \in \sH_A \otimes \sH_B \otimes \sH_C$, consisting of system $A$ with the $n$ qubits that $V_0$ kept, and the systems $B$ and $C$ obtained by applying the isometry to the qubits $E_0$ received from $V_0$. Since the isometry is independent of $\theta$\,---\,$E_0$ needs to decide on it before he finds out what $\theta$ is\,---\,so is the state $\ket{\psi_{ABC}}$. It is clear that in order to pass the position verification test the adversaries must win a restricted version of the game $\mathsf{G}_{\BB}^{\times n}$.\footnote{The extra restriction comes from the fact that they have no access to the qubits kept by $V_{0}$ and so the reduced state on those must be fully mixed. It turns out that this restriction does not affect the optimal winning probability.} Therefore, the probability $\varepsilon$ that $x'_0 = x = x'_1$ is bounded by $\pwin(\mathsf{G}_{\BB}^{\times n})$. Our Theorem~\ref{th:perfect} thus concludes the proof.
\end{proof}

The security of the position verification scheme can be immediately extended to adversaries that hold a linear amount of shared entanglement.

\begin{corollary}\label{cor:posver}
The above position verification scheme is \smash{$d \cdot (\frac{1}{2}\!+\!\frac{1}{2\sqrt{2}})^n$}-sound against adversaries $(E_0,E_1)$ that share an arbitrary (possibly entangled) state $\eta_{E_{0} E_{1}}$, such that $\dim \eta_{E_{0} E_{1}} = d$, at the time they receive $Q$ and $\theta$, respectively.
\end{corollary}
Thus, for any $\alpha$ strictly smaller than \smash{$\log(\frac{1}{2}\!+\!\frac{1}{2\sqrt{2}})$}, for instance for $\alpha = 0.2$, the position verification scheme has exponentially small soundness error (in $n$) against adversaries that hold at most $\alpha n$ pre-shared entangled qubits.

Corollary~\ref{cor:posver} is an immediate consequence of Proposition~\ref{prop:posver} above and of Lemma~V.3 of~\cite{beigi11}. The latter states that $\varepsilon$-soundness with no entanglement implies $(d \cdot \varepsilon)$-soundness for adversaries that pre-share a $d$-dimensional state. This follows immediately from the fact that the pre-shared state can be extended to a basis of the $d$-dimensional state space, and the uniform mixture of all these basis states gives a non-entangled state (namely the completely mixed state). As a consequence, applying the attack, which is based on the entangled state, to the setting with no entanglement, reduces the success probability by at most a factor of $d$.

By the results on imperfect guessing (see Section~\ref{sec:fixedError}), at the price of correspondingly weaker parameters, the above results extend to a noise-tolerant version of the scheme, where it is sufficient for $x'$ to be {\em close}, rather than equal, to $x$ for $V_0$ and $V_1$ to accept.

\section{Application III: Entropic Uncertainty Relation}\label{sc:ucr}

Let $\rho$ be an arbitrary state of a qubit and $\Theta$ a uniformly random bit. Then, we may consider the min-entropy
of $X$, where $X$ is the outcome when $\rho$ is measured in either one of two bases with overlap $c$, as determined by $\Theta$.
For this example, it is known that~\cite{deutsch83,schaffner07}
\begin{align}
  H_{\min}(X|\Theta)_{\rho} \geq -\log \frac{1 + \sqrt{c}}{2} \label{ucr-c1}.
\end{align}
A similar relation follows directly from results by Maassen and Uffink~\cite{maassen88}, namely
\begin{align}
  H_{\min}(X|\Theta)_{\rho} + H_{\max}(X|\Theta)_{\rho} \geq -\log c \, , \label{ucr-c2}
\end{align}
where, $H_{\max}$ denotes the R\'enyi entropy~\cite{renyi61} of order $\frac{1}{2}$.

Recently, entropic uncertainty relations have been generalized to the case where the party guessing $X$
has access to quantum side information~\cite{berta10}. However, note that a party that is maximally entangled with the state of
the system to be measured can always guess the outcome of $X$ by applying an appropriate measurement (depending on $\Theta$) on the entangled state.
Thus, there cannot be any non-trivial state-independent bound on the entropies above conditioned on quantum side information.
Nonetheless, if two disjoint quantum memories are considered, the following generalization of~\eqref{ucr-c2} was
shown. For an arbitrary tripartite state $\rho_{ABC}$ and $X$ measured on $A$ as prescribed above, one
finds~\cite{tomamichel11}
\begin{align}
  H_{\min}(X|B \Theta)_{\rho} + H_{\max}(X| C \Theta)_{\rho} \geq -\log c \,. \label{ucr-q2}
\end{align}
In the following, we show a similar generalization of the uncertainty relation in~\eqref{ucr-c1} to quantum side information.


\begin{theorem}
  \label{th:ucr}
  Let $\rho_{ABC}$ be a quantum state and $\Theta$ a uniformly random bit. Given two
  POVMs $\{ F_{x}^{0} \}$ and $\{ F_{x}^{1} \}$ with overlap $c := \max_{x, z} \norm[\big]{\sqrt{F_{x}^{0}} \sqrt{F_{z}^{1}}}^{2}$, we find
  \begin{align*}
   \pg(X|B \Theta)_{\rho} + \pg(X|C \Theta)_{\rho} \leq 1 + \sqrt{c}
  \end{align*}
  and
  \begin{align*}
    H_{\min}(X|B \Theta)_{\rho} + H_{\min}(X|C \Theta)_{\rho} \geq -2 \log \frac{1 + \sqrt{c}}{2} ,
  \end{align*}
  where the quantities are evaluated for the post-measurement state
  \begin{align}
    \rho_{XBC\Theta} = \sum_{x,\theta} \frac{1}{2}\, \ketbra{x}_{X} \otimes \Tr_A \big( (F_x^\theta \otimes \id_{BC}) \rho_{ABC} \big) \otimes \ketbra{\theta}_{\Theta} . \label{eq:pm}
  \end{align}
\end{theorem}

\begin{proof}
  First, recall that the min-entropy is defined as (cf.~Eq.~\eqref{eq:hmin-theta})
  \begin{align*}
    2^{-H_{\min}(X|B\Theta)_{\rho}} = \pg(X|B\Theta)_{\rho} = \max_{ \{ P_x^{\theta} \}}
    \sum_{x,\theta} p_{x,\theta}\, \tr{\rho_B^{x,\theta} P_x^{\theta}}
    = \max_{ \{ P_x^{\theta} \}} \frac{1}{2} \sum_{x,\theta} \Tr\big( \rho_{AB} ( F_x^{\theta} \otimes P_x^{\theta} ) \big) ,
  \end{align*}
  where we used the fact that the post-measurement states given by~\eqref{eq:pm} satisfy
  $p_{x,\theta}\, \rho_{BC}^{x,\theta} = \frac{1}{2} \Tr_{A}
  \big( F_x^{\theta} \rho_{ABC} \big)$.

  In the following argument, we restrict ourselves to the case where the optimal guessing
  strategies for the min-entropy, $\{ P_x^\theta \}$ for Bob and
  $\{ Q_x^\theta \}$ for Charlie, are projective. To see that this is sufficient, note that we
  can always embed the state $\rho_{XBC}$ into a larger system $\rho_{XB'C'}$
  such that the optimal POVMs on $B$ and $C$ can be diluted into an equivalent projective measurement strategy
  on $B'$ and $C'$, respectively. The data-processing inequality of the min-entropy then tells us that
  $H_{\min}(X|B\Theta) \geq H_{\min}(X|B'\Theta)$ and $H_{\min}(X|C\Theta) \geq H_{\min}(X|C'\Theta)$,
  i.e., it is sufficient to find a lower bound on the smaller quantities,
  for which the optimal strategy is projective.
 
  For an arbitrary state $\rho_{ABC}$ and optimal projective POVMs $\{ P_x^\theta \}$ and
  $\{ Q_x^\theta \}$, we have
  \begin{align*}
    2^{-H_{\min}(X|B \Theta)_{\rho}} + 2^{-H_{\min}(X|C \Theta)_{\rho}}
    &= \frac{1}{2} \sum_{x,\theta} \Tr \Big( \rho_{ABC} \big( F_x^{\theta} \otimes P_x^{\theta}
    \otimes \id_C + F_x^{\theta} \otimes \id_B \otimes Q_x^\theta \big)  \Big) \\
    &\leq \frac{1}{2}  \norm[\bigg]{\sum_{x,\theta} F_x^{\theta} \otimes P_x^{\theta}
    \otimes \id_C + F_x^{\theta} \otimes \id_B \otimes Q_x^\theta} \,.
  \end{align*}

  We now upper-bound this norm. First, we rewrite
  \begin{align*}
    \norm[\bigg]{\sum_{x,\theta} F_x^{\theta} \otimes P_x^{\theta}
    \otimes \id_C + F_x^{\theta} \otimes \id_B \otimes Q_x^\theta}
     = \norm[\bigg]{\sum_{i,\theta} A_i^{\theta}} \leq \norm[\big]{A_0^0 + A_1^1} +
      \norm[\big]{A_1^0 + A_0^1} ,
  \end{align*}
  where $A_0^{\theta} = \sum_x F_x^{\theta} \otimes P_x^{\theta} \otimes \id_C$ and $A_1^{\theta} =
  \sum_x F_x^{\theta} \otimes \id_{B} \otimes Q_x^{\theta}$ are projectors.
  Applying Lemma~\ref{lm:mykittaneh} twice then yields
  \begin{align*}
    \norm[\big]{A_0^0 + A_1^1} + \norm[\big]{A_1^0 + A_0^1}
    &\leq 2 + \norm[\Big]{\sqrt{A_0^0} \sqrt{A_1^1}} + \norm[\Big]{\sqrt{A_0^1} \sqrt{A_1^0}}\\
		&\leq 2 + 2 \max_{x, z} \norm[\big]{\sqrt{F_x^{0}} \sqrt{F_{z}^{1}}} \leq 2 + 2 \sqrt{c},
  \end{align*}
  where we used that $\norm[\big]{A_i^{\theta}} \leq 1$. Hence,
  \begin{align*}
    2^{-H_{\min}(X|B \Theta)_{\rho}} + 2^{-H_{\min}(X|C \Theta)_{\rho}} = 
    \pg(X|B \Theta)_{\rho} + \pg(X|C \Theta)_{\rho}
    \leq 1 + \sqrt{c} .
  \end{align*}
  and, using the relation between arithmetic and geometric mean, we finally get
  \begin{align*}
    2^{-H_{\min}(X|B \Theta)_{\rho}}2^{- H_{\min}(X|C \Theta)_{\rho}}
    \leq \left( \frac{1+\sqrt{c}}{2} \right)^2 ,
  \end{align*}
  which implies the statement of the lemma after taking the logarithm on both sides.
\end{proof}

Note that, for $n$ measurements, each in a basis chosen uniformly at random, the above result still only guarantees one bit of uncertainty.
In fact, an adaptation of the proof of Theorem~\ref{th:ucr} yields the bound
\begin{align*}
  H_{\min}(X^n|B \Theta^n) + H_{\min}(X^n|C \Theta^n) \geq -2 \log \frac{1 + \sqrt{c^n}}{2} \,.
\end{align*}
This bound can be approximately achieved using a state that is maximally entangled between $A$
and $B$ with probability $\frac{1}{2}$ and maximally entangled between $A$ and $C$ otherwise. This construction
ensures that both conditional min-entropies are low and we thus cannot expect a stronger result.
This is in stark contrast to the situation with classical side information in~\eqref{ucr-c1} and
the alternative uncertainty relation~\eqref{ucr-q2}, where the
lower bound on the uncertainty can be shown to scale linearly in
$n$ (cf.~\cite{wehner09,tomamichel11}). Due to this restriction, we expect that the applicability of
Theorem~\ref{th:ucr} to quantum cryptography is limited.

\section{Conclusion}\label{sc:conc}

We introduce the notion of a monogamy-of-entanglement game, and we show a general parallel repetition theorem. For a BB84-based example game, we actually show {\em strong} parallel repetition, and that a non-entangled strategy is sufficient to achieve the optimal winning probability. Our results have various applications to quantum cryptography. 

It remains open to understand which monogamy-of-entanglement games satisfy strong parallel repetition. Another open question is whether (or in what cases) a {\em concentration theorem} holds, which states that with high probability the fraction of won executions in a parallel repetition cannot be much larger than the probability of winning a single execution. 

With respect to our applications, an interesting open problem is to increase the noise level that can be tolerated for one-sided device-independent security of BB84. 
It is not clear at all that the rather low noise level of $1.5\%$ we obtain in our analysis is inherent; this may very well be an artifact of our technique. {Finally, it would be interesting to extend our analysis to incorporate channel losses following the work of Branciard \emph{et al.}~\cite{branciard12}. As suggested there,  we expect that such an analysis would reveal a higher tolerance for losses as compared to fully DI QKD.}

\subsubsection*{Acknowledgements}

We thank Renato Renner for early discussions and Niek J.~Bouman for bringing this problem to the attention of some of us.
MT, JK and SW are funded by the Ministry of Education (MOE) and National Research Foundation Singapore, as well as MOE Tier 3 Grant "Random numbers from quantum processes" (MOE2012-T3-1-009).

\appendix

\section{Pure Strategies are Sufficient}
\label{app:pure}

\begin{lemma}
  \label{lm:pure}
  In the supremum over strategies in~\eqref{eq:win}, it is sufficient to consider pure strategies.
\end{lemma}
\begin{proof}
Given any strategy ${\cal S} = \{\rho_{ABC}, P_x^{\theta}, Q_x^{\theta} \}$ for a game $\mathsf{G}$, we construct a pure strategy $\tilde{\cal S} = \{\proj{\tilde{\varphi}}, \tilde{P}_x^{\theta}, \tilde{Q}_x^{\theta} \}$ with $\pwin(\mathsf{G}, \tilde{\cal S}) = \pwin(\mathsf{G}, {\cal S})$. 
First, it is clear that purifying $\rho_{ABC}$, with a purifying register that is appended to $C$, does not change the value of $\pwin(\mathsf{G}, {\cal S})$. Hence, we may assume that $\rho_{ABC}$ is already pure: $\rho_{ABC} = \proj{\varphi}$. In this case, $\pwin(\mathsf{G}, {\cal S})$ simplifies to 
$$
\pwin(\mathsf{G}, {\cal S}) = \sum_{x,\theta} \frac{1}{|\Theta|} \bra{\varphi} (\proj{x^\theta} \otimes P_x^\theta \otimes Q_x^\theta) \ket{\varphi} \, .
$$
Let $\sH_X$ be a Hilbert space of dimension $|\cX|$ and with basis $\{ \ket{x} \}_x$, and let
$\ket{\psi_0}$ be an arbitrary, fixed vector in $\sH_X$. We now set  $\ket{\tilde{\varphi}} = \ket{\varphi} \otimes \ket{\psi_0} \in \sH_A \otimes \sH_B \otimes \sH_C \otimes \sH_X$ as well as $\tilde{P}_x^{\theta} = U_\theta^\dagger (\id_B \otimes \proj{x}) U_\theta$, where $U_\theta \in \olin{\sH_B \otimes \sH_X}$ is a Neumark dilation unitary that maps
$$
\ket{\psi}\otimes \ket{\psi_0} \mapsto \sum_{x \in \cX} \sqrt{P_x^{\theta}}\ket{\psi} \otimes \ket{x}
$$
for every $\ket{\psi} \in \sH_B$.  
Then, $\tilde{P}_x^{\theta}$ is indeed a projection and hence $\tilde{P}_x^{\theta} = (\tilde{P}_x^{\theta})^\dagger \tilde{P}_x^{\theta}$, and%
\footnote{It is implicitly understood that $\tilde{P}_x^{\theta}$ only acts on the $BX$ part of $\ket{\tilde{\varphi}}$, and similarly for $U_\theta$ etc. } 
$$
\tilde{P}_x^{\theta} \ket{\tilde{\varphi}} = U_\theta^\dagger (\id_B \otimes \proj{x}) U_\theta \big(\ket{\varphi} \otimes \ket{\psi_0} \big) = U_\theta^\dagger \sqrt{P_x^{\theta}} \ket{\varphi} \otimes \ket{x} \, .
$$
Similarly, we define the projection $\tilde{Q}_x^{\theta}$ (and extend the state $\ket{\tilde{\varphi}}$). It then follows immediately that $\pwin(\mathsf{G}, \tilde{\cal S}) = \pwin(\mathsf{G}, {\cal S})$.
\end{proof}

\section{Equivalence of QKD Security Definitions}
\label{app:secdefs}

To prove security of a protocol, it is sufficient to show that the security criterion is satisfied by a state close to the true output state of the protocol. This is due to the following Lemma.

\begin{lemma}\label{lemma:SecDefs}
Let $\rho_{XB},\, \tilde{\rho}_{XB} \in \dens(\sH_X \otimes \sH_B)$ be two CQ states with $X$ over $\cal X$. Also, let $\lambda: {\cal X} \to \{0,1\}$ be a predicate on $\cal X$ and $\Lambda = \lambda(X)$, and let $\tau_X \in \dens(\sH_X)$ be arbitrary. Then
$$
\Pr_{\rho}[\Lambda] \cdot \Delta(\rho_{XB|\Lambda}, \,\tau_X \otimes \rho_{B|\Lambda}) \leq 5\Delta(\rho_{XB}, \,\tilde{\rho}_{XB}) + \Pr_{\tilde{\rho}}[\Lambda] \cdot \Delta(\tilde{\rho}_{XB|\Lambda}, \,\tau_X \otimes \tilde{\rho}_{B|\Lambda}) \, .
$$   
\end{lemma}

\begin{proof}
We set $\delta := \Delta(\rho_{XB}, \tilde{\rho}_{XB})$. 
From $\Delta(\rho_{XB}, \tilde{\rho}_{XB}) = \delta$ it follows in particular that the two distributions $P_X$ and $\tilde{P}_X$ are $\delta$-close, and thus that the state 
$$
\sigma_{XB} := \Pr_{\rho}[\Lambda] \cdot \tilde{\rho}_{XB|\Lambda} +  \Pr_{\rho}[\neg\Lambda] \cdot \tilde{\rho}_{XB|\neg \Lambda}
$$
is $\delta$-close to $\tilde{\rho}_{XB}$, and hence $2\delta$-close to $\rho_{XB}$, where $\neg \Lambda$ is the negation of the event $\Lambda$. 
Since $\Lambda$ is determined by $X$, we can write 
$$
\Delta(\rho_{XB},\sigma_{XB}) = \Pr_{\rho}[\Lambda] \cdot \Delta(\rho_{XB|\Lambda}, \,\tilde{\rho}_{XB|\Lambda}) +  \Pr_{\rho}[\neg\Lambda] \cdot \Delta(\rho_{XB|\neg \Lambda}, \,\tilde{\rho}_{XB|\neg \Lambda}) \, ,
$$
from which it follows that $\Pr_{\rho}[\Lambda] \cdot \Delta(\rho_{XB|\Lambda}, \,\tilde{\rho}_{XB|\Lambda}) \leq 2 \delta$, and, by tracing out $X$, also that $\Pr_{\rho}[\Lambda] \cdot \Delta(\rho_{B|\Lambda},\tilde{\rho}_{B|\Lambda}) \leq 2 \delta$. 
We can now conclude that 
\begin{align*}
\Pr_{\rho}[\Lambda] \cdot \Delta(\rho_{XB|\Lambda}, \,\tau_X \otimes \rho_{B|\Lambda}) &\leq 4\delta + \Pr_{\rho}[\Lambda] \cdot \Delta(\tilde{\rho}_{XB|\Lambda}, \,\tau_X \otimes \tilde{\rho}_{B|\Lambda})\\
&\leq 5\delta + \Pr_{\tilde{\rho}}[\Lambda] \cdot \Delta(\tilde{\rho}_{XB|\Lambda}, \,\tau_X \otimes \tilde{\rho}_{B|\Lambda})\, ,
\end{align*}
which proves the claim. 
\end{proof}

\bibliographystyle{arxiv}
\bibliography{library,stephslib}

\end{document}